%% file: asfs.tex
\documentclass[pre,twocolumn,showpacs,superscriptaddress,preprintnumbers,floatfix]{revtex4}

\usepackage{ifpdf}
\usepackage{hyperref}
\usepackage{dcolumn}
\usepackage{url}
\usepackage{amsmath}
\usepackage{amscd}
\usepackage{amsfonts}
\usepackage{amssymb}
\usepackage{bm}   
\usepackage{bbm}   
\usepackage{verbatim}
\usepackage{stmaryrd}
\usepackage{amsthm}

\ifpdf
  \usepackage[pdftex]{graphicx}         
\else
  \usepackage[dvips]{graphicx}
\fi

\input{cmechabbrev}

\input{shortcuts}

\begin{document}

\title{Asymptotic Synchronization for Finite-State Sources}

\author{Nicholas F. Travers}
\email{ntravers@math.ucdavis.edu}
\affiliation{Complexity Sciences Center}
\affiliation{Mathematics Department}

\author{James P. Crutchfield}
\email{chaos@cse.ucdavis.edu}
\affiliation{Complexity Sciences Center}
\affiliation{Mathematics Department}
\affiliation{Physics Department\\
University of California at Davis,\\
One Shields Avenue, Davis, CA 95616}
\affiliation{Santa Fe Institute\\
1399 Hyde Park Road, Santa Fe, NM 87501}

\date{\today}

\bibliographystyle{unsrt}

\begin{abstract}

We extend a recent synchronization analysis of exact finite-state sources to
nonexact sources for which synchronization occurs only asymptotically. Although
the proof methods are quite different, the primary results remain
the same. We find that an observer's average uncertainty in the source state
vanishes exponentially fast and, as a consequence, an observer's average
uncertainty in predicting future output converges exponentially fast to the
source entropy rate. 

\end{abstract}

\pacs{
02.50.-r  
89.70.+c  
05.45.Tp  
02.50.Ey  
}
\preprint{Santa Fe Institute Working Paper 10-11-029}
\preprint{arxiv.org:arXiv:1011.1581 [nlin.CD]}

\maketitle


\section{Introduction}
\vspace{-.1in}

In Ref. \cite{Trav10a} we analyzed the synchronization process for exact 
\eMs, where the observer may come to know the internal state of the machine
with certainty after only a finite number of measurements. Here, we examine the 
case of nonexact \eMs, where the observer may only synchronize to the machine's
state asymptotically. Although the analysis differs, the behavior is
qualitatively similar to the exact case in the sense that an observer (on
average) synchronizes to a nonexact machine exponentially fast. That is, there
exist constants $K > 0$ and $0 < \alpha < 1$ such that the average state
entropy $\AvgUncertainty(L) \leq K \alpha^L$, for all $L \in \N$.

Our development is organized as follows.
Section \ref{sec:Defns} briefly reviews the synchronization problem and
provides the essential definitions for our results.
Section \ref{sec:Picture} presents an intuitive picture of the synchronization
process, using it to derive a formula for $\phi(w)$, the
conditional state distribution induced by a word $w$.
Section \ref{sec:EntropyRate} establishes a formula for the entropy rate of a
finite-state \eM.
Section \ref{sec:AvAsymSync} uses the entropy-rate formula to prove
the existence of averaged asymptotic synchronization.
Section \ref{sec:SyncRateThm} builds on this result to prove
our main theorem---the Nonexact Machine Synchronization Theorem.
Section \ref{sec:EntropyConv} uses this theorem  to show that, for any nonexact \eM,
the state entropy $\AvgUncertainty(L)$ vanishes exponentially fast and the
length-$L$ entropy-rate approximation $\hmu(L)$ converges exponentially fast
to the machine's entropy rate. Finally, Sec. \ref{sec:Concl} summarizes
our results and examines possible extensions.

\vspace{-.1in}
\section{Background}
\label{sec:Defns}
\vspace{-.1in}

This section lays out the necessary definitions and background for our results.
For a more thorough introduction the reader is referred to Ref. \cite{Trav10a},
where a similar but more detailed presentation is given.

\vspace{-.2in}
\subsection{Machines}
\vspace{-.1in}

\begin{Def}
\emph{Hidden Markov machine}:
\label{Def:HMM}
A \emph{finite-state edge-label hidden Markov machine (HMM)} consists of
\begin{enumerate}
\setlength{\topsep}{0mm}
\setlength{\itemsep}{0mm}
\item a finite set of states
	$\CausalStateSet = \{\causalstate_1, ... , \causalstate_N \}$,
\item a finite alphabet of symbols $\MeasAlphabet$, and
\item a set of $N$ by $N$ symbol-labeled transition matrices $T^{(\ms)}$,
	$\ms \in \MeasAlphabet$,
	where $T^{(\ms)}_{ij}$ is the probability of transitioning from
	state $\causalstate_i$ to state $\causalstate_j$ on symbol $\ms$.
	The corresponding internal state-to-state transition matrix is denoted 
	$T = \sum_{\ms \in \MeasAlphabet} T^{(\ms)}$.
\end{enumerate}
A hidden Markov machine can be depicted as a directed graph with labeled edges.
The nodes are the states $\{\causalstate_1, ... , \causalstate_N \}$ and for
all $\ms,i,j$ with $T^{(\ms)}_{ij} > 0$, there is an edge from state
$\causalstate_i$ to state $\causalstate_j$ labeled $p|\ms$ for the symbol
$\ms$ and transition probability $p = T^{(\ms)}_{ij}$. We require that the
transition matrices $T^{(\ms)}$ be such that this graph is strongly connected.
\end{Def}

A hidden Markov machine $M$ generates a stationary process
$\Process = (\MS_L)_{L \geq 0}$ as follows. Initially, $M$ starts in some state
$\causalstate_{i^*}$ chosen according to the stationary distribution $\pi$ over
machine states---the distribution satisfying $\pi T = \pi$.  It then picks an
outgoing edge according to their relative transition probabilities
$T^{(\ms)}_{i^*j}$, emits the symbol $\ms^*$ labeling this
edge, and follows the edge to a new state $\causalstate_{j^*}$. The next
output symbol and state are consequently chosen in a similar fashion, and this
procedure is repeated indefinitely. 

We denote by $\CS_0, \CS_1, \CS_2, \ldots$ the random variables (RVs) for the
sequence of machine states visited and by $\MS_0, \MS_1, \MS_2, \ldots$ the RVs
for the associated sequence of output symbols generated. The sequence of states
$(\CS_L)_{L \geq 0}$ is a Markov chain with transition kernel $T$. However,
the stochastic process we
consider is not the sequence of states, but rather the associated sequence of
outputs $(\MS_L)_{L \geq 0}$, which is not normally Markov. We assume that
an observer of the process sees the sequence of outputs, but does not have
direct access to the machine's ``hidden'' internal states.

\paragraph*{Example: Even Process Machine}

Figure~\ref{fig:EvenProcess} gives an HMM for the Even Process, a machine that
has been studied extensively \cite{EvenProcessMerge}. Its name derives
from the feature that in its output there are always an even number of $1$s
between consecutive $0$s. The transition matrices are:
\begin{align*}
T^{(0)} & =
	\left(
	\begin{array}{cc}
		p & 0 \\
		0 & 0 \\
	\end{array}
	\right) ~,
	\nonumber \\
T^{(1)} & =
	\left(
	\begin{array}{cc}
	0 & 1-p \\
	1 & 0 \\
	\end{array}
	\right) ~.
\end{align*}

\begin{figure}[h]
\centering
\includegraphics[scale=0.6]{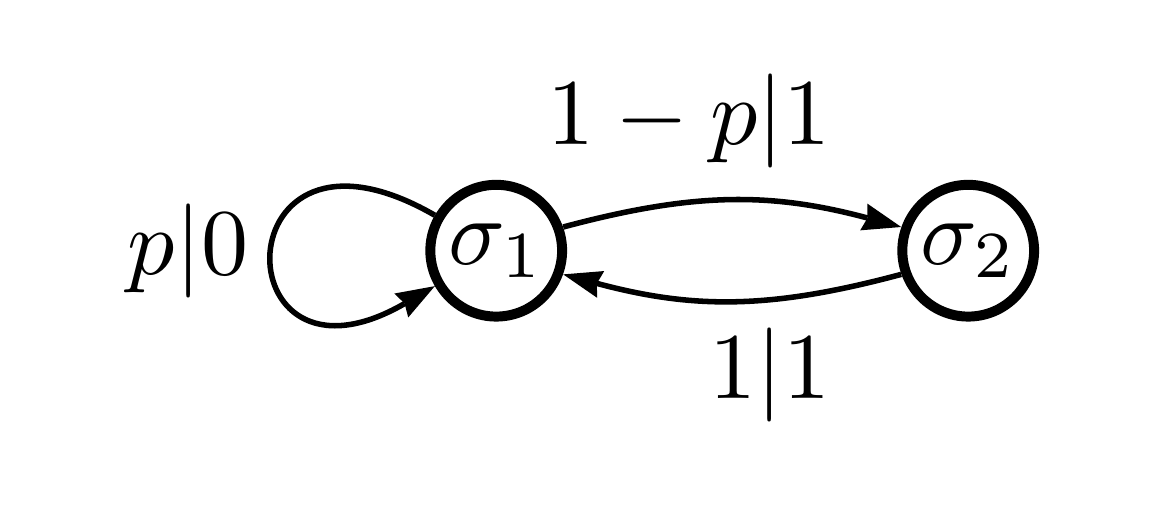}
\caption{A hidden Markov machine (the \eM) for the Even Process.  The
  transitions denote the probability $p$ of generating symbol $\ms$ as $p|\ms$.
  }
\label{fig:EvenProcess}
\end{figure}

The following notation will be used for sequences of output RVs:
\begin{enumerate}
\setlength{\topsep}{0mm}
\setlength{\itemsep}{0mm}
\item $\Future = \MS_0 \MS_1 \ldots$ ,
\item $\Future^L = \MS_0 \MS_1 \ldots \MS_{L-1} $ , \mbox{ and }
\item $\Future_t^L = \MS_t \MS_{t+1} \ldots \MS_{t+L-1}$ .
\end{enumerate}

\begin{Def}
A \emph{finite-state \eM} is a finite-state edge-label hidden Markov machine
with the following properties:
\begin{enumerate}
\setlength{\topsep}{0mm}
\setlength{\itemsep}{0mm}
\item \emph{Unifilarity}: For each state $\causalstate_k \in \CausalStateSet$
	and each symbol $\ms \in \MeasAlphabet$ there is at most one
	outgoing edge from state $\causalstate_k$ labeled with symbol $\ms$.
\item \emph{Probabilistically distinct states}: For each pair of distinct states
	$\causalstate_k, \causalstate_j \in \CausalStateSet$ there exists some
	finite word $w = \ms_0 \ms_1 \ldots \ms_{L-1}$ such that:
	\begin{equation*}
	\Prob(\Future^L = w|\CS_0 = \causalstate_k)
		\not= \Prob(\Future^L = w|\CS_0 = \causalstate_j) ~.
	\end{equation*}
\end{enumerate} 
\end{Def}

\paragraph*{Example (continued)}
The Even Process machine given above is also an \eM. It is clearly unifilar,
and $\causalstate_1$ can generate the symbol $0$ whereas $\causalstate_2$
cannot, so the states are probabilistically distinct.

\begin{Rem}
\EMs\ were originally defined in Ref. \cite{Crut88a} as hidden Markov
machines whose states, known as \emph{causal states}, were the equivalence
classes of infinite pasts $\past$ with the same probability distribution over
futures $\future$. This ``history machine'' definition is, in fact, equivalent
to the ``generating machine'' definition presented above in the finite-state
case. Although, this is not immediately apparent. Formally, it follows
from the synchronization results established here and in Ref. \cite{Trav10a}.
\end{Rem}

We now provide the definitions for two extensions of an \eM\ $M$ that are
necessary for our proofs later on: the edge machine $M_{edge}$ and the power
machine $M^n$. In what follows:
\begin{enumerate}
\setlength{\topsep}{0mm}
\setlength{\itemsep}{0mm}
\item $\Prob(\ms|\causalstate_k) \equiv \Prob(\MS_0 = \ms|\CS_0 = \causalstate_k)$,
\item $\Prob(w|\causalstate_k) \equiv \Prob(\Future^{|w|} = w|\CS_0 = \causalstate_k)$,
\item $I(\ms,k,j)$ denotes the indicator function of the transition from state
$\causalstate_k$ to state $\causalstate_j$ on symbol $\ms$, and
\item $I(w,k,j)$ denotes the indicator function of the transition from state
$\causalstate_k$ to state $\causalstate_j$ on the word $w$.
\end{enumerate}
That is, $I(\ms,k,j) = 1$ if $\causalstate_k \goesonx \causalstate_j$ and $0$
otherwise; $I(w,k,j) = 1$ if  $\causalstate_k \goesonw \causalstate_j$ and $0$
otherwise.

\begin{Def}
For an \eM\ $M$, the corresponding \emph{edge machine} $M_{edge}$ is the Markov
chain whose states are the outgoing edges of $M$. That is, the states are the
pairs $(x,\causalstate_k)$ such that $\Prob(\ms|\causalstate_k) > 0$, and the
transition probabilities are defined as:
\begin{align*}
\Prob((\ms,\causalstate_k) \rightarrow (y, \causalstate_j))
	= \Prob(y|\causalstate_j) I(\ms,k,j) ~.
\end{align*}
\end{Def}

A sequence of $M_{edge}$ states visited by the Markov chain corresponds to a
sequence of edges visited by the original machine $M$. The process
$\Process_{edge}$ generated by $M_{edge}$ can be thought of as the
\emph{bi-process} $(\MS_L, \CS_L)_{L \geq 0}$ generated by the original machine 
$M$ as it moves from state to state generating symbols. Note that since $M$'s
graph is strongly connected, $M_{edge}$'s graph is as well. Hence, the
edge-label Markov chain is irreducible and has a unique stationary distribution
$\pi_{edge}$. See Fig. \ref{fig:EdgePowerMachines}(top).

\begin{Def}
Let $M$ be an \eM, and let $n$ be relatively prime to the period $p$ of $M$'s
graph. The \emph{power machine} $M^n$ is defined to be the \eM\ with the states
of $M$, output symbols which are length-$n$ words generated by $M$, and
transition probabilities given by:
\begin{align*}
\Prob(\causalstate_k \goesonw \causalstate_j)
	= \Prob(w|\causalstate_k) I(w,k,j) ~.
\end{align*}
The power machine $M^n$ generates the same process as the original machine $M$,
but over length-$n$ blocks.
\end{Def}

Note that since $M$ is by definition unifilar with probabilistically distinct
states, $M^n$ is also necessarily unifilar with probabilistically distinct
states. Furthermore, it can be shown that for $n$ relatively prime to
$p = \mathrm{period}(M)$ the graph of $M^n$ is strongly connected. Therefore, for
$n$ relatively prime to $p$, $M^n$ is indeed an \eM\ for the process
$\Process^n$. See Fig. \ref{fig:EdgePowerMachines}(bottom).

\begin{figure}[h]
\centering
\includegraphics[scale=0.5]{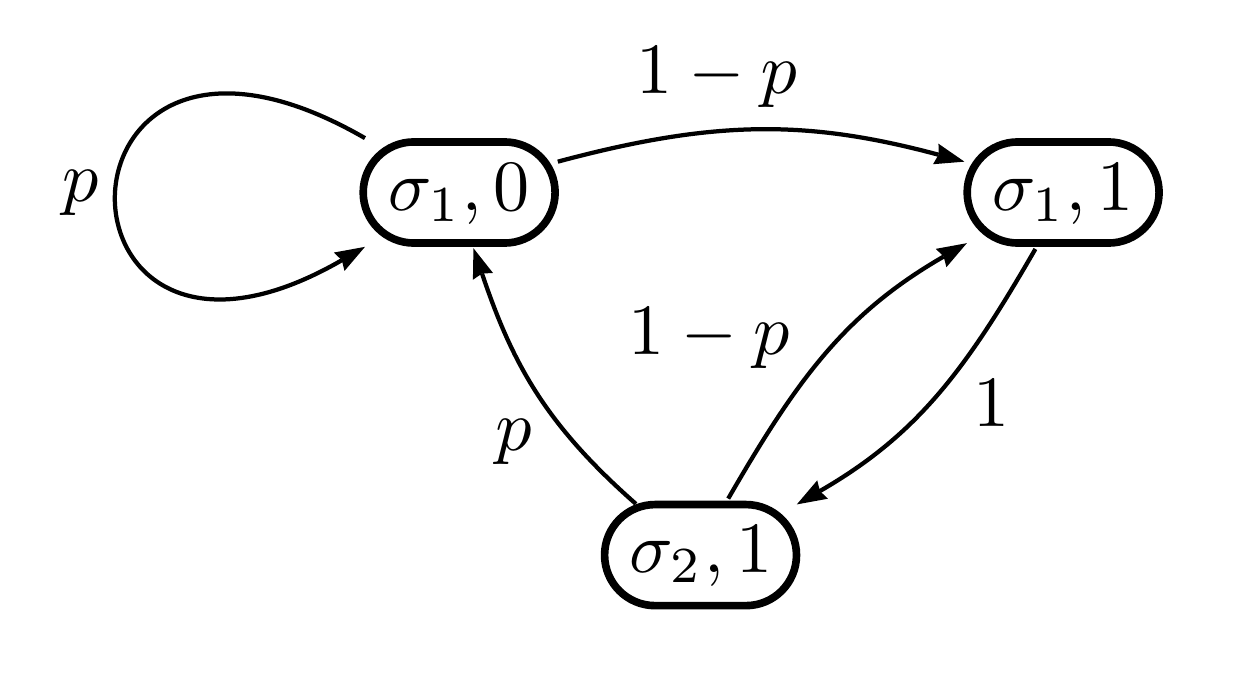} 
\includegraphics[scale=0.5]{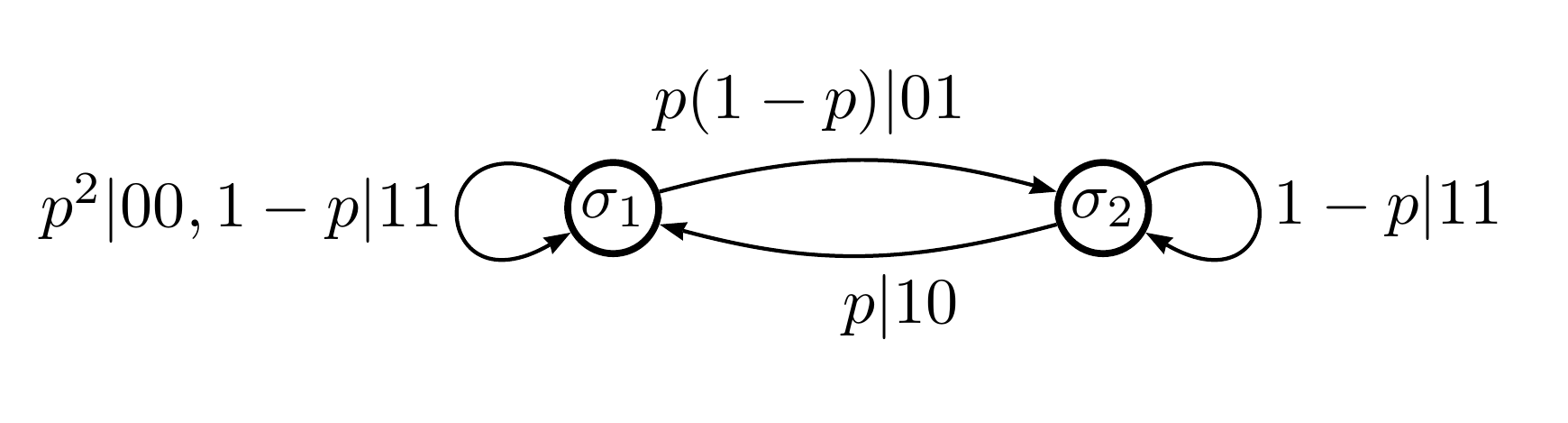}
\caption{Examples of $M_{edge}$ (top) and $M^2$ (bottom) for the Even Process
  \eM\ $M$.
  }
\label{fig:EdgePowerMachines}
\end{figure}

\begin{Def}
For an \eM\ $M$ the \emph{minimum distinguishing length} $L^*$ is the shortest
length $L$ such that the probability distributions over futures $\Future^L$ of
length $L$ are distinct for each pair of distinct states $\causalstate_k$ and
$\causalstate_j$:
\begin{align*}
L^* \equiv \min \{ & L: \Prob(\Future^L|\CS_0 = \causalstate_k)
  \not= \Prob(\Future^L|\CS_0 = \causalstate_j) ~, \\ 
  & \mbox{for all } k \not= j \} ~.
\end{align*} 
If a machine $M$ has a minimum distinguishing length $L^*$, we also say
that $M$ has \emph{length-$L^*$ future distinguishable states}. 
\end{Def}

Note that $L^*$ must be finite for any \eM, since \eMs\ have probabilistically
distinct states, and that, for $n \geq L^*$ and relatively prime to
$p = \mathrm{period}(M)$, $M^n$ is an \eM\ with a minimum distinguishing length
of 1.

\vspace{-0.2in}
\subsection{Synchronization} 
\label{sec:Sync}
\vspace{-0.1in}

Although we assume that our observer is not able to directly see the \eM's
internal state ($\CS_L$), it is able to see the output symbols generated by
the machine (the $\MS_L$'s). Thus, the observer may attempt to infer the
internal machine state through observations of the output. We are
interested in studying the procedure by which the observer synchronizes to the
machine's state through these observations. Due to unifilarity, we
know that if an observer is able to completely synchronize to the machine's
internal state at some time $T > 0$, it remains synchronized for all future
times $T^\prime \geq T$. For simplicity, we assume that the initial state is
chosen according to the stationary distribution $\pi$, so that the process generated
by the machine is stationary, and also that the observer has knowledge of this fact. 

For a word $w$ of length $L$ generated by the machine let
$\phi(w) \equiv \Prob(\CausalStateSet|w)$ be the observer's
\emph{belief distribution} as to the current state of the machine after
observing $w$. That is, 
\begin{align*}
\phi(w)_k	& = \Prob(\CS_L = \causalstate_k | \Future^L=w) \\
		& \equiv \Prob(\CS_L = \causalstate_k | \Future^L=w, \CS_0 \sim \pi) ~.
\end{align*}
And, define the observer's uncertainty in the machine state after observing
$w$ as:
\begin{align*}
u(w) & = H[\phi(w)] \\
  & = H[\CS_L|\Future^L = w] ~.
\end{align*} 

Let $\LM$ denote the set of all finite words that $M$ can generate, $\LLM$
the set of all length-$L$ words it can generate, and $\LiM$ the set of 
all infinite sequences $\future = \ms_0 \ms_1 ...$ that it can generate.

\begin{Def}
A word $w \in \LM$ is a \emph{synchronizing word} (or \emph{sync word)} for $M$
if $u(w) = 0$; that is, if the observer knows the current state of the machine 
with certainty after observing $w$.
\end{Def}

We denote the set of $M$'s infinite synchronizing sequences as $\SYN(M)$ and the set of $M$'s infinite weakly synchronizing sequences as $\WSYN(M)$:
\begin{align*}
& \SYN(M) = \{ \future \in \mathcal{L}_{\infty}(M) : u(\future^L) = 0 \mbox{ for some } L\} ,~ and \\
& \WSYN(M) = \{ \future \in \LiM : u(\future^L) \rightarrow 0 \mathrm{~as~} L \rightarrow \infty \} ~.
\end{align*} 

\begin{Def}
\label{Def:ExactSync}
An \eM\ M is \emph{exactly synchronizable} (or simply \emph{exact}) if
$\Prob(\SYN(M)) = 1$; that is, if the observer synchronizes to almost every
(a.e.) sequence the machine generates in finite time.
\end{Def}

\begin{Def}
\label{Def:AsymptoticSync}
An \eM\ M is \emph{asymptotically synchronizable} if $\Prob(\WSYN(M)) = 1$;
that is, if the observer's uncertainty in the machine state vanishes asymptotically
for a.e. sequence the machine generates. 
\end{Def}

\paragraph*{Examples:}
\begin{itemize}
\item The Even Process \eM\ is an exact machine. Any word containing a $0$
	is a sync word for this machine, and almost every $\future$ generated by
	this machine contains at least one $0$.
\item The ABC machine (Fig. \ref{fig:ABC}) is not exactly synchronizable, but
	it is asymptotically synchronizable. 
\end{itemize}

\begin{figure}[h]
\centering
\includegraphics[scale=0.6]{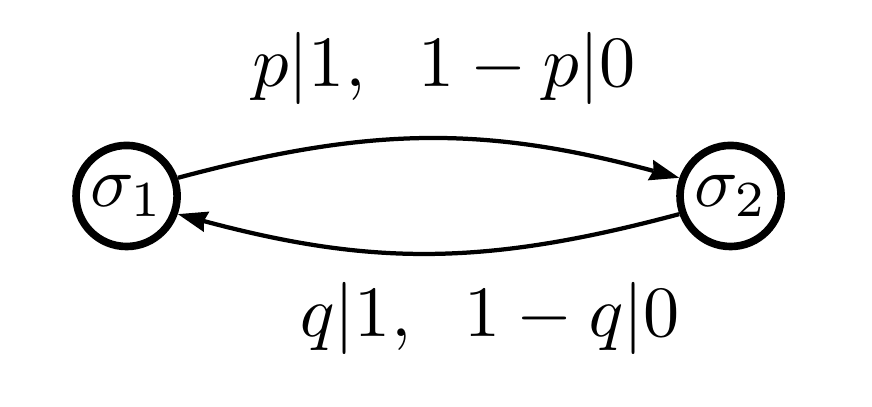}
\caption{The Alternating Biased Coin (ABC) machine: The process it generates
  can be thought of as alternately flipping two coins of different biases,
  $p \neq q$.
  }
\label{fig:ABC}
\end{figure}

We note that any machine with a single state is necessarily exact since the observer 
is synchronized before observing any output. However, the synchronization question in 
this case is moot. Thus, when discussing exact or nonexact machines,
we will always assume $N \geq 2$. Also, since any finite word $w \in \L(M)$ is contained 
in a.e. infinite sequence $\future$ an \eM\ $M$ generates, we know that a machine $M$ is exact 
if (and only if) it has some sync word $w$ of finite length. 

One final important quantity to monitor during synchronization is the observer's
average uncertainty in the machine state after seeing a length-$L$ block of
output.

\begin{Def}
The observer's \emph{average state uncertainty} at time L is:
\begin{align*}
\AvgUncertainty(L) & \equiv H[\CS_L|\Future^L] \\
  & = \sum_{\{ \future^L \}} \Prob(\future^L) u(\future^L) ~.
\end{align*}
\end{Def}

\vspace{-0.3in}
\subsection{Prediction} 
\vspace{-0.1in}

A process's intrinsic randomness is measured by its entropy rate and that, in
turn, determines how well an observer can predict its behavior.

\begin{Def} 
The \emph{block entropy} $H(L)$ for a stationary process $\Process$ is:
\begin{align*}
H(L) & \equiv H[\Future^L] \\
  & = - \sum_{\{ \future^L \} } \Prob(\future^L) \log_2 \Prob(\future_L) ~.
\end{align*}
\end{Def}

\begin{Def} 
The \emph{entropy rate} $\hmu$ is the asymptotic average entropy per symbol:
\begin{align*}
\hmu & \equiv \lim_{L \to \infty} \frac{H(L)}{L} \\
  & = \lim_{L \to \infty} H[\MS_L|\Future^L] ~.
\end{align*}
\end{Def}

\begin{Def} 
Its \emph{length-$L$ approximation} is:
\begin{align*}
\hmu(L) & \equiv H(L) - H(L-1) \\
  & = H[\MS_{L-1}|\Future^{L-1}] ~.
\end{align*}
That is, $\hmu(L)$ is the observer's average uncertainty in the next symbol to
be generated after observing the first $L-1$ symbols.
\end{Def}

For any stationary process, $\hmu(L)$ monotonically decreases to the limit
$\hmu$ \cite{Cove06a}. However, the form of convergence depends on the 
process. The lower the value of $\hmu$ a process has, the better an observer's
predictions of the process will be asymptotically. The faster $\hmu(L)$ converges to
$\hmu$, the faster an observer's predictions will reach this optimal asymptotic level. 
Since we are often interested in making predictions after only a finite sequence 
of observations, the source's true entropy rate $\hmu$, as well as the rate of 
convergence of $\hmu(L)$ to $\hmu$, are both important properties.

Now, for an $\epsilon$-machine, an observer's prediction of the next output
symbol is a direct function of the probability distribution over machine
states induced by the previously observed symbols. That is, 
\begin{align*}
\Prob(\MS_L = \ms & | \Future^L = \future^L) \\
  & = \sum_k \Prob(\ms|\causalstate_k)
  \Prob(\CS_L = \causalstate_k | \Future^L = \future^L) ~.
\end{align*}
Hence, the better an observer knows the machine state at the current time, the
better it can predict the next symbol the machine generates. And, on average,
the closer $\AvgUncertainty(L)$ is to $0$, the closer $\hmu(L)$ is to $\hmu$.
Therefore, the rate of convergence of $\hmu(L)$ to $\hmu$ for an
$\epsilon$-machine is closely related to the average rate of synchronization.
This is one of the primary motivations for studying the synchronization
problem. 

\vspace{-0.2in}
\section{An Intuitive Picture}
\label{sec:Picture}
\vspace{-0.1in}

In this section we present an intuitive picture of the synchronization process
and use it to derive a formula for the conditional state distribution $\phi(w)$.
The basic idea is illustrated schematically in Fig. \ref{fig:SyncPicture} for a
hypothetical $5$-state machine.

\begin{figure}[h]
\centering
\includegraphics[scale=0.43]{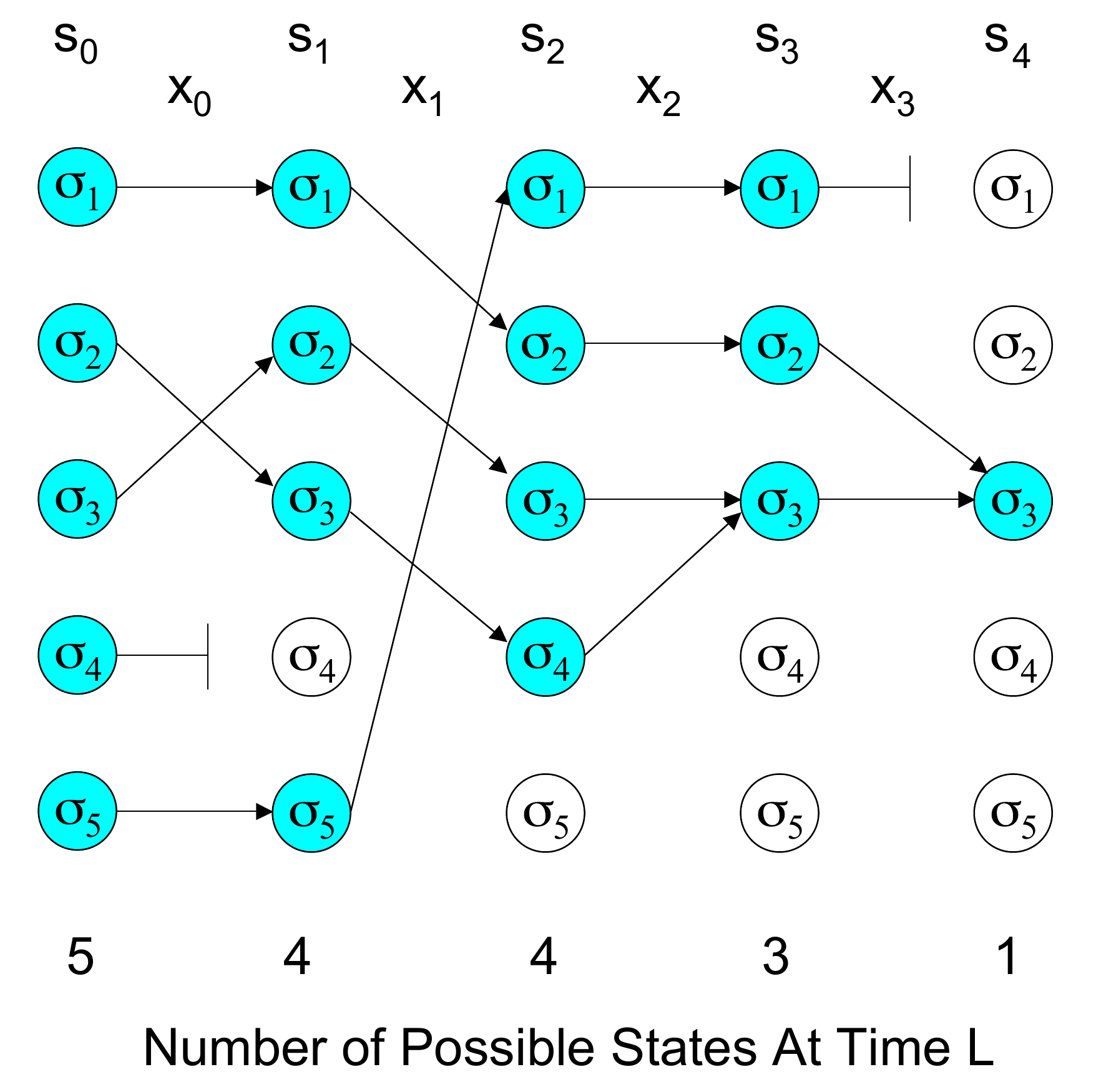}
\caption{Synchronization illustrated for a $5$-state machine.}
\label{fig:SyncPicture}
\end{figure}

Initially, the observer does not know the machine state $\CS_0$, so all five
states $\{\causalstate_1, \causalstate_2, \causalstate_3, \causalstate_4,
\causalstate_5 \}$ are possible. After seeing the first symbol $\ms_0$,
there are only four possibilities for
$\CS_1$---$\{\causalstate_1, \causalstate_2, \causalstate_3, \causalstate_5
\}$---since only four of the five states may generate this symbol. After seeing
the second symbol $\ms_1$, a different set of four states is
possible---$\{\causalstate_1, \causalstate_2, \causalstate_3,\causalstate_4 \}$.
After seeing the third symbol $\ms_2$, there
are only three possibilities
$\{\causalstate_1, \causalstate_2, \causalstate_3 \}$ for $\CS_3$, since two of
the state paths merge on seeing the third symbol. Finally, after seeing the
fourth symbol $\ms_3$, two more state paths merge and another dies, so
there is only one possibility $\{\causalstate_3\}$ for $\CS_4$.
The observer has synchronized.

The transition function $\delta(\causalstate_k,\ms)$ is defined by the relation
$\causalstate_k \goesonx \delta(\causalstate_k,\ms)$ and the word transition
function $\delta(\causalstate_k,w)$ by the relation
$\causalstate_k \goesonw \delta(\causalstate_k,w)$.
In general, for each possible initial state $\causalstate_k$ and each
$\future^L$ there is a state path $p^k$ following $\causalstate_k$ under
$\future^L$. That is, $p^k = p^k_0, p^k_1, \ldots, p^k_L$, where
$p^k_0 = \causalstate_k$, $p^k_1 = \delta(p^k_0,\ms_0)$,
$p^k_2 = \delta(p^k_1,\ms_1)$, and so on. An observer synchronizes exactly once
all, but one, of these paths have either merged or died.

If a machine is not exactly synchronizable, then it is impossible for all
paths to merge or die and, at any finite time $L$, there are at least
two possible nonmerged paths remaining. However, it is still possible for an
observer to synchronize to such a machine asymptotically. To understand how
this happens we need to know the relative probabilities of being in each of
the possible remaining states at a given time. In general, the probability of
starting in state $\causalstate_k$ and generating the word $\future^L$ is:
\begin{align*}
\Prob(p^k) & \equiv \Prob(\CS_0 = \causalstate_k, \Future^L = \future^L) \\
  & = \pi_k \cdot \Prob(\future^L|\causalstate_k) ~.
\end{align*}
These probabilities will be exactly $0$ if and only if the path $p^k$ dies by
the $L_{th}$ symbol. Typically, however, all these probabilities decay---in
fact, decay exponentially fast---as $L \rightarrow \infty$. For nonexact
synchronization, we are concerned not with absolute path probabilities, but
with their relative or normalized probabilities. The probability of ending up
in state $\causalstate_j$ at time $L$ is simply the sum of the normalized
probabilities of all paths ending up in state $\causalstate_j$. That is, for
any word $w = \future^L$ we have:
\begin{align}
\phi(w)_j & \equiv \Prob(\CS_L = \causalstate_j |\Future^L = w) \nonumber \\
  & = \frac{\Prob(\CS_L = \causalstate_j, \Future^L = w)}{\Prob(\Future^L=w)} \nonumber \\
  & = \frac{\sum_k \pi_k \cdot \Prob(w|\causalstate_k) \cdot I(w,k,j)}
	{\sum_i \pi_i \cdot \Prob(w | \causalstate_i)} \nonumber \nonumber \\
  & =  \frac{\sum_k \Prob(p^k) \cdot I(w,k,j)}{\sum_i \Prob(p^i)} ~.
\label{eq:ProbStateGivenW}
\end{align}

For nonexact asymptotic synchronization, then, the important quantities to
consider are the relative probabilities of all paths that never merge or die.
If a nonexact machine is asymptotically synchronizable, then for a.e. $\future$ 
there must be some state $\cs_k$ such that the ratio of the path probabilities:
\begin{align*}
\frac{\Prob (p^k(\future^L))}{\Prob(p^j(\future^L))}
  \rightarrow \infty~,
\end{align*}
as $L \rightarrow \infty$, for any path $p^j$ which does not eventually merge with
$p^k$ or die. Since $\pi_k / \pi_j$ is bounded for all states $\causalstate_k$ and 
$\causalstate_j$, the initial state is unimportant for asymptotic synchronization. 
The question is whether, on average, the transition probabilities for one path are 
greater than those of the other. If, on average, the transition probabilities for path
$p^k$ are $c$ ($ > 1$) times as likely as the transition probabilities for path
$p^j$ then, for large $L$, $\Prob(p^k)/\Prob(p^j) \sim c^L$. Intuitively, this
is why synchronization occurs exponentially fast. Establishing this, as we will
see, requires some care, however.
 

\section{The Entropy Rate Formula}
\label{sec:EntropyRate}

In this section we derive a formula for the entropy rate of a finite-state
\eM. Although an analogous expression has been previously established in
similar contexts (see, e.g., Ref. \cite{Shan62}), we provide a derivation
as well for completeness. The proof presented here is also somewhat simpler
than the original in Ref. \cite{Shan62}. A proof quite similar to ours for
unifilar Moore hidden Markov models (as opposed to the edge-label or
Mealy models we use) is given in Ref. \cite{Mass75a}.   

\begin{Prop}
For any \eM\ $M$,
\begin{align}
\hmu & = H[\MS_0|\CS_0] \nonumber \\
  & \equiv \sum_k \pi_k h_k ~,
\end{align}
where $h_k = H[\MS_0|\CS_0 = \causalstate_k]$. 
\end{Prop}

\begin{proof}
We establish the bounds from above and below separately.

\emph{Upper bound}: $\hmu \leq H[\MS_0|\CS_0]$. We calculate directly:
\begin{align*}
\hmu 	& \equiv \lim_{L \to \infty} \frac{H[\Future^L]}{L} \\
		& \leq \lim_{L \to \infty} \frac{H[\CS_0, \Future^L]}{L} \\
		& \stackrel{(a)}{=} \lim_{L \to \infty} \frac{H[\CS_0] + \sum_{i= 0}^{L-1} H[\MS_i|\CS_0,\Future^i]}{L} \\
		& \stackrel{(b)}{=} \lim_{L \to \infty} \frac{H[\CS_0] + \sum_{i= 0}^{L-1} H[\MS_i|\CS_i]}{L} \\
		& \stackrel{(c)}{=} \lim_{L \to \infty} \frac{H[\CS_0] + H[\MS_0|\CS_0] \cdot L}{L} \\
		& = H[\MS_0|\CS_0] ~,
\end{align*}
where step (a) follows from the chain rule, step (b) from unifilarity, and step
(c) from stationarity.

\emph{Lower bound}: $\hmu \geq H[\MS_0|\CS_0]$. We have:
\begin{align*}
\hmu 	& \equiv \lim_{L \to \infty} \frac{H[\Future^L]}{L} \\
		& \geq \lim_{L \to \infty} \frac{H[\Future^L|\CS_0]}{L} \\
		& \stackrel{(a)}{=} \lim_{L \to \infty} \frac{\sum_{i= 0}^{L-1} H[\MS_i|\CS_0,\Future^i]}{L} \\
		& \stackrel{(b)}{=} \lim_{L \to \infty} \frac{\sum_{i= 0}^{L-1} H[\MS_i|\CS_i]}{L} \\
		& \stackrel{(c)}{=} \lim_{L \to \infty} \frac{H[\MS_0|\CS_0] \cdot L}{L} \\
		& = H[\MS_0|\CS_0] ~,
\end{align*}
where again step (a) follows from the chain rule, step (b) from unifilarity,
and step (c) from stationarity.
\end{proof}

\section{Averaged Asymptotic Synchronization} 
\label{sec:AvAsymSync}

The entropy rate formula says that (on average) an observer predicts 
asymptotically just as well as if it knew the machine state exactly:
\begin{align*}
\hmu & = \lim_{L \to \infty} H[\MS_L|\Future^L] \\
  & = H[\MS_0|\CS_0] ~.
\end{align*}
Intuitively, this suggests that an observer's average uncertainty
$\AvgUncertainty(L)$ in the machine state must vanish asymptotically. That is,
we should have:
\begin{equation}
\lim_{L \to \infty} \AvgUncertainty(L) = 0 ~. \nonumber
\end{equation}
These ideas are made rigorous below with a convexity argument.

The following notation will be used:
\begin{itemize}
\setlength{\topsep}{0mm}
\setlength{\itemsep}{0mm}
\item Let $p_k \equiv \Prob(\MS_0|\CS_0 = \causalstate_k)$ and
	$p_w \equiv \Prob(\MS_0|\CS_0 \sim \phi(w))$.
\item Let $h_k \equiv H[p_k]$ (as above), $h_w \equiv H[p_w]$, and
	$\th_w \equiv \sum_k \phi(w)_k h_k$.
\item Let $\AeL \equiv \{ w \in \LLM: u(w) < \epsilon \}$ and
	$\AeL^c \equiv \LLM/ \AeL$, the complement of $\AeL$.
\end{itemize}

We note that for any word $w$:
\begin{equation}
p_w = \sum_k p_k \phi(w)_k ~, 
\end{equation}
and, hence, by the concavity of the entropy function $H[\cdot]$:
\begin{align}
\label{Eq:h_vs_th}
h_w	& = H[p_w] \nonumber \\
	& = H \left[ \sum_k p_k \phi(w)_k \right] \nonumber \\
	& \geq \sum_k \phi(w)_k H[p_k] \nonumber \\
	& = \sum_k \phi(w)_k h_k \nonumber \\
	& = \th_w ~.
\end{align}
Also, for any length $L$: 
\begin{align}
\label{Eq:hmuLp1_sum}
\hmu(L+1) 	& = H[\MS_L|\Future^L] \nonumber \\
			& = \sumw \Prob(w) h_w ~,
\end{align}
and
\begin{align}
\label{Eq:hmu_sum}
\hmu 	& = \sum_k \pi_k h_k \nonumber \\
  		& = \sum_k \left( \sumw \Prob(w) \phi(w)_k \right) h_k \nonumber \\
  		& = \sumw \Prob(w) \sum_k \phi(w)_k h_k \nonumber \\
  		& = \sumw \Prob(w) \th_w ~.
\end{align}

\begin{Prop} For any finite-state \eM\ $M$:
\begin{equation}
\label{Prop:AvgUncertaintyConvergence}
\lim_{L \to \infty} \AvgUncertainty(L) = 0 ~.
\end{equation}
\end{Prop} 

\begin{proof}
We first prove the statement under the assumption that $M$ has a minimum 
distinguishing length $L^* = 1$. We then use this result to establish the 
general case. 

\emph{Case (i): $M$ has minimum distinguishing length $L^* = 1$. }
The proof is by contradiction. If $\AvgUncertainty(L) \not\rightarrow 0$, then  
there must be some $\epsilon > 0$ for which $\Prob(\AeL^c) \not \rightarrow 0$. 
Hence, there exists some $\delta > 0$ and a subsequence $(L_i)_{i=1}^{\infty}$ 
of the $L$s such that $\Prob(\AeLi^c) \geq \delta$, for all $i$.

Let $\Delta$ be the unit simplex in $\R^N$:
\begin{equation*} 
\Delta = \left\{ \phi \in \R^N: \sum_k \phi_k = 1
  \mbox{ and } \phi_k \geq 0, \mbox{ for all } k \right\} ~,
\end{equation*}
and let:
\begin{equation*}
\Delta_{\epsilon} = \left\{ \phi \in \Delta: H[\phi] \geq \epsilon \right\} ~.
\end{equation*}
Define $f: \Delta_{\epsilon} \rightarrow \R$ by:
\begin{equation*}
f(\phi) = H\left[\sum_k \phi_k p_k\right] - \sum_k \phi_k H[p_k] ~,
\end{equation*}
so that, for any word $w$, $f(\phi(w)) = h_w - \th_w$.

Then, with respect to $\norm{\cdot}_1$, $f(\phi)$ is a continuous function on
$\Delta_{\epsilon}$ and $\Delta_{\epsilon}$ is a compact set. Therefore, we
know $f$ obtains its minimum $f^*$ at some point $\phi^* \in \Delta_{\epsilon}$.

Since $M$ has a minimum distinguishing length $L^* = 1$ and the entropy function
$H[\cdot]$ is strictly concave, we know $f(\phi) > 0$ for all
$\phi \in \Delta_{\epsilon}$. In particular, $f(\phi^*) = f^* > 0$. 

Hence, for each $i$ we have:
\begin{align}
\label{Eq:gap}
& \hmu(L_i+1) - \hmu \nonumber \\
& \stackrel{(a)}{=} \sum_{w \in \L_{L_i}(M)} \Prob(w) h_w - \sum_{w \in \L_{L_i}(M)} \Prob(w) \th_w \nonumber \\ 
& = \sum_{w \in \L_{L_i}(M)} \Prob(w) \cdot (h_w - \th_w) \nonumber \\
&  \stackrel{(b)}{\geq} \sum_{w \in \AeLi^c} \Prob(w) \cdot (h_w - \th_w) \nonumber \\
& \geq \Prob(\AeLi^c) \cdot f^* \nonumber \\
& \geq \delta f^* ~.
\end{align}
Step (a) follows from Eqs. (\ref{Eq:hmuLp1_sum})
and (\ref{Eq:hmu_sum}) and step (b) follows from Eq. (\ref{Eq:h_vs_th}).
Equation (\ref{Eq:gap}) implies that $\hmu(L) \not\rightarrow \hmu$, which is
a contradiction. Hence, we know $\lim_{L \to \infty} \AvgUncertainty(L) = 0$.

\emph{Case (ii): $M$ has minimum distinguishing length $L^* > 1$.}
Take $n \geq L^*$ and relatively prime to the period $p$ of $M$'s graph, so that 
$M^n$ is an \eM\ with a minimum distinguishing length of $1$. Let $Y_L$ be the RV for the
$L_{th}$ output symbol generated by the machine $M^n$, let $R_L$ be the RV for
$M^n$'s $L_{th}$ state, and let $\V(L) = H[R_L |\vY^L]$. Note that for any $L$:
\begin{equation}
\label{Eq:BlockEquality}
\V(L) = \AvgUncertainty(nL).
\end{equation}

Now, for a contradiction assume
$\lim_{L \to \infty} \AvgUncertainty(L) \not= 0$. Then, since
$\AvgUncertainty(L)$ is monotonically decreasing, we know there exists some
$\epsilon > 0$ such that $\AvgUncertainty(L) \geq \epsilon$, for all $L$. Thus,
by Eq. (\ref{Eq:BlockEquality}), we know that $\V(L) \geq \epsilon$ for all
$L$ as well, so $\V(L) \not\rightarrow 0$. However, since $M^n$ has a minimum
distinguishing length of 1, by case (i) above we know that $\V(L)$ must go to zero. 
This contradiction implies that $\lim_{L \to \infty} \AvgUncertainty(L) = 0$. 
\end{proof}
  
\section{The Nonexact Machine Synchronization Theorem}
\label{sec:SyncRateThm}

In this section we prove our primary result, the Nonexact Machine
Synchronization Theorem. This extends the weak asymptotic synchronization
result of Sec. \ref{sec:AvAsymSync} to show that synchronization occurs
exponentially fast for nonexact machines. The statement is quite analogous to
the Exact Machine Synchronization Theorem given in Ref. \cite{Trav10a}.
Essentially, it says that, except on a set of words $\future^L$ of exponentially
small probability, an observer's uncertainty after observing $\future^L$ is
exponentially small.  

The following notation will be used. $\Phi_L \equiv \phi(\Future^L)$ is the
random variable for the belief distribution over states induced by the first
length-$L$ word the machine generates, and $\CSb_L$ is the most likely state in
$\Phi_L$ (if a tie the lowest numbered state is taken).
$P_L \equiv \Prob(\CSb_L)$ is the probability of the most likely state in the
distribution $\Phi_L$, and $Q_L \equiv \Prob(\mbox{NOT } \CSb_L)$ is the
combined probability of all other states in the distribution $\Phi_L$. 
For example, if $\Phi_L = (0.2,0.7,0.1)$, then $\CSb_L = \causalstate_2$,
$P_L = 0.7$, and $Q_L = 0.3$. Realizations are denoted $\phi_L$, $\sb_L$,
$p_L$, and $q_L$, respectively. We also define $U_L = H[\Phi_L]$ 
and $u_L = H[\phi_L]$. 

\begin{The} (Nonexact Machine Synchronization Theorem)
For any nonexact \eM\ $M$,
\begin{enumerate}
\setlength{\topsep}{0mm}
\setlength{\itemsep}{0mm}
\item There exist constants $K_1 > 0$ and $0 < \alpha_1 < 1$ such that:
\begin{align*}
\Prob(Q_L > \alpha_1^L) \leq K_1 \alpha_1^L ~, \mbox{ for all $L \in \N$}.
\end{align*}
\item There exist constants $K_2 > 0$ and $0 < \alpha_2 < 1$ such that:
\begin{align*}
\Prob(U_L > \alpha_2^L) \leq K_2 \alpha_2^L ~, \mbox{ for all $L \in \N$}.
\end{align*}
\end{enumerate}
\label{the:NESyncRate}
\end{The}

The proof strategy is as follows. We first take a power machine $M^n$ of
the machine $M$ with $\AvgUncertainty(n) = \epsilon \ll 1$, and prove the
theorem for the power machine. We then use the exponential convergence of the
power machine to establish the theorem in general with a subsequence-type
argument. 

The following lemma on large deviations of Markov chains will be critical. 

\begin{Lem} 
Let $Z_0, Z_1, ... $ be a finite-state, irreducible Markov chain, with state
set $R = \{r_1, ... , r_n \}$ and equilibrium distribution
$\rho = (\rho_1, ... , \rho_n)$. Let $F: R \rightarrow \R$, $Y_L = F(Z_L)$,
and $\Yb_L = \frac{1}{L} (Y_0 + ... + Y_{L-1})$. Define
$\mu_F = \Ex_{\rho}(F) = \sum_k \rho_k F(r_k)$. Then, for any $\epsilon > 0$, 
there exist constants $K > 0$ and $0 < \alpha < 1$ such that, for any state
$r_k$:
\begin{align*}
\Prob \left( |\Yb_L - \mu_F| > \epsilon|\CS_0 = r_k \right)
  \leq K \alpha^L ~, \mbox{ for all } L \in \N ~.
\end{align*}
\label{lem:FiniteIrreducibleConverge}
\end{Lem}

\vspace{-0.3in}
\begin{proof}
A similar statement (with more explicit values of the constants) is given in
Ref. \cite{Glyn02a} for a general class of Markov chains, which includes all
finite-state, irreducible, aperiodic chains. The result stated here follows
directly for finite-state, irreducible, aperiodic chains, and can be extended
to the periodic case by considering length $p$-blocks, where $p$ is the chain's
period. 
\end{proof}

\begin{Rem}
Note that since the deviation bound holds conditionally on any initial state
$r_k$, it also holds conditionally on any distribution over the initial state by
linearity. In particular, we apply this lemma assuming $Z_0 \sim \rho$.
\end{Rem}

Let us denote:
\begin{align*}
\Prob(\ms,\causalstate_k) & = \Prob(\CS_0 = \causalstate_k, \MS_0 = \ms) ~, \\
\Prob(\ms|\causalstate_k) & = \Prob(\MS_0 = \ms | \CS_0 = \causalstate_k) ~, \\
\Prob(\causalstate_k|\ms) & = \Prob(\CS_0 = \causalstate_k | \MS_0 = \ms) ~, \mbox{ and } \\
\causalstate_{\max,\ms} & = \mbox{argmax } \Prob(\causalstate_k | \ms) ~,
\end{align*}
where again the lowest numbered state is chosen in the case of a tie for
$\causalstate_{\max,\ms}$. Also, for any $\ms$ and $\causalstate_j$ with
$\Prob(\ms|\causalstate_j) > 0$, let us define:
\begin{align*}
\CausalStateSet_{\ms,j} & = \{ \causalstate_k \in \CausalStateSet:
  \Prob(\ms|\causalstate_k) > 0~,~
  \delta(\causalstate_k,\ms) \not= \delta(\causalstate_j,\ms) \} ~, \\
g(\ms, \causalstate_j) & = \max_{\causalstate_k \in \CausalStateSet_{\ms,j}}
  \Prob(\ms|\causalstate_k) ~, \mbox{ and } \\
f(\ms, \causalstate_j) & = \max_{\causalstate_k \in \CausalStateSet_{\ms,j}}
  \Prob(\causalstate_k|\ms) ~. 
\end{align*}
Note that $g(\ms, \causalstate_j)$ and $f(\ms, \causalstate_j)$ are both always
strictly positive for nonexact \eMs. And, also, that for any joint length-$L$
realization $(\future^L,\futurestates^L)$:

\begin{align}
\frac{p_L}{q_L} \geq
\frac{\Prob(\csr_0)}{\Prob(\mbox{NOT $\csr_0$})}
\prod_{i=0}^{L-1} \frac{\Prob(\ms_i|\csr_i)}{g(\ms_i,s_i)} ~,
\label{Eq:SyncNotSyncProduct}
\end{align}
by Eq. (\ref{eq:ProbStateGivenW}). Here, $\Prob(\csr_0) = \pi_k$ is the
stationary probability of the state $\csr_0 = \causalstate_k$ and
$\Prob(\mbox{NOT $\csr_0$}) = 1 - \Prob(\csr_0)$. 

Using Lemma \ref{lem:FiniteIrreducibleConverge} we now prove our desired
theorem under the (relatively strong) assumption that:
\begin{align}
\Ex_{\pi_{edge}} \left\{ \log_2 \left(
  \frac{\Prob(\MS_0 | \CS_0)}{g(\MS_0, \CS_0)} \right) \right\} > 0 ~.
\label{eq:StrongAssumption}
\end{align}
This assumption will later be satisfied for some power machine $M^n$.

\begin{Lem} 
Let $M$ be a nonexact \eM\ satisfying Eq. (\ref{eq:StrongAssumption}). Then $:$
\begin{enumerate}
\setlength{\topsep}{0mm}
\setlength{\itemsep}{0mm}
\item There exist constants $K_1 > 0$ and $0 < \alpha_1 < 1$ such that:
\begin{align*}
\Prob(Q_L > \alpha_1^L) \leq K_1 \alpha_1^L ~, \mbox{ for all } L \in \N.
\end{align*}
\item There exist constants $K_2 > 0$ and $0 < \alpha_2 < 1$ such that:
\begin{align*}
\Prob(U_L > \alpha_2^L) \leq K_2 \alpha_2^L ~, \mbox{ for all } L \in \N.
\end{align*}
\end{enumerate}
\label{lem:ConstantsAssumingExpectation}
\end{Lem}

\begin{proof} We first prove Claim 1 and then use this to show Claim 2.

\emph{Proof of Claim 1}:
Consider the edge-label Markov process $\Process_{edge}$ generated by the
edge machine $M_{edge}$. Let $Z_L = (\MS_L, \CS_L)$ denote the RV for the 
$L_{th} ~ M_{edge}$-state and let:
\begin{align*}
Y_L & = F(Z_L) \\
  & = \log_2 \left( \frac{\Prob(\MS_L|\CS_L)}{g(\MS_L,\CS_L)} \right) ~.
\end{align*}
We assume, of course, that
$(\MS_0,\CS_0) \sim \pi_{edge}$ or, equivalently, $\CS_0 \sim \pi$. 

By hypothesis, $\mu_F = \Ex_{\pi_{edge}}(F) = C > 0$. Take $\epsilon = C/2$.
By Lemma \ref{lem:FiniteIrreducibleConverge}, there exist constants $B_1 > 0$
and $0 < \eta_1 < 1$ such that:
\begin{align*}
\Prob(|\Yb_L - \mu_F| > \epsilon) \leq B_1 \eta_1^L ~,
  \mbox{ for all } L \in \N ~.
\end{align*}
Thus, for any $L$:
\begin{align*}
\Prob\left( \sum_{i=0}^{L-1} \log_2 \left(
  \frac{\Prob(\MS_i|\CS_i)}{g(\MS_i,\CS_i)}  \right) < \frac{C}{2} L \right) 
	& = \Prob(\Yb_L < C/2) \\
	& = \Prob(\Yb_L < \mu_F - \epsilon) \\
	& \leq \Prob(|\Yb_L - \mu_F| > \epsilon) \\
	& \leq B_1 \eta_1^L ~.
\end{align*}
Now, let $\vz^L = (\future^L, \futurestates^L)$ be any \emph{typical sequence}, i.e.:
\begin{align*}
\sum_{i=0}^{L-1} \log_2 \left( \frac{\Prob(\ms_i|\ms_i)}{g(\ms_i,\ms_i)}
  \right) \geq \frac{C}{2} L. 
\end{align*}
Taking logarithms of Eq. (\ref{Eq:SyncNotSyncProduct}) we find:
\begin{align*}
\log_2 \left( \frac{p_L}{q_L} \right)	
	& \geq \log_2 \left( \frac{\Prob(\csr_0)}{\Prob(\mbox{NOT $\csr_0$})}  \right) + 
	\sum_{i=0}^{L-1} \log_2 \left( \frac{\Prob(\ms_i|\csr_i)}{g(\ms_i,\csr_i)} \right) \nonumber \\ 	& \geq  \beta + \frac{C}{2}L~,				
\end{align*}
where $\beta \equiv \min_k \log_2 \left( \frac{\Prob(\CS_0 = \causalstate_k)}{\Prob(\CS_0 \not= \causalstate_k)} \right)$.
Or, equivalently:
\begin{align*}
\frac{p_L}{q_L} \geq 2^{\beta} \cdot 2^{\frac{C}{2}L} = B_2 \eta_2^L  ~,
\end{align*}
where $B_2 \equiv 2^{\beta} > 0$ and $\eta_2 \equiv 2^{C/2} > 1$. Thus:
\begin{align*}
q_L \leq \frac{q_L}{p_L} \leq B_3 \eta_3^L,
\end{align*}
where $B_3 \equiv 1 / B_2 > 0$ and $\eta_3 \equiv 1 / \eta_2 < 1$. 
Since this holds for any typical sequence $(\future^L,\futurestates^L)$ we
have, for each $L$:
\begin{align*}
\Prob(Q_L > B_3 \eta_3^L) \leq B_1 \eta_1^L.
\end{align*}
And, therefore, for any $1 > \alpha_1 > \max \{\eta_1,\eta_3\}$ there exists
some $K_1 = K_1(\alpha_1)$ sufficiently large that:
\begin{align*}
\Prob(Q_L > \alpha_1^L) \leq K_1 \alpha_1^L ~, \mbox{ for all } L \in \N ~.
\end{align*}

\emph{Proof of Claim 2}:
By Claim 1 we know there exist constants $K_1 > 0$ and $0 < \alpha_1 < 1$
such that:
\begin{align*}
\Prob(Q_L > \alpha_1^L) \leq K_1 \alpha_1^L ~, \mbox{ for all } L \in \N ~.
\end{align*}
Let us define:
\begin{align*}
V_L^+ & = \{ \future^L: q_L > \alpha_1^L \}
  \mbox{ and } V_L^- = \{ \future^L : q_L \leq \alpha_1^L\} ~.
\end{align*}

Take $L_1$ sufficiently large that $1 - \alpha_1^L \geq 1/2$, for all
$L \geq L_1$. Note that the first-order Taylor expansion about $x = 1$
of $\log_2(1 - \alpha_1^L)$ is $- \log_2(e)\alpha_1^L + O(\alpha_1^{2L}) \approx -1.44 \alpha_1^L + O(\alpha_1^{2L})$. 
Thus, there exists some $L_2 \in \N$ such that $|\log_2(1 - \alpha_1^L)| \leq 2 \alpha_1^L$ for all
$L \geq L_2$. Take $L_0 = \max \{L_1,L_2\}$. 

Then, for any $L \geq L_0 $ and any $\future^L \in V_L^-$, we have:
\begin{align}
  & H[\CS_L|\future^L] \nonumber \\
  & \stackrel{(a)}{\leq} H\left[\left(1 - \alpha_1^L, \frac{\alpha_1^L}{N-1}, \dots , \frac{\alpha_1^L}{N-1} \right) \right] \nonumber \\ 
  & = - \left[ (1 - \alpha_1^L) \log_2(1 - \alpha_1^L) + \alpha_1^L \log_2\left( \frac{\alpha_1^L}{N-1} \right) \right] \nonumber \\
  & = -(1 - \alpha_1^L) \log_2(1-\alpha_1^L) \nonumber \\
  &~~~~~~ - \alpha_1^L L \log_2(\alpha_1) + \alpha_1^L \log_2(N-1) \nonumber \\
  & \stackrel{(b)}{\leq} (1 - \alpha_1^L) 2 \alpha_1^L - \alpha_1^L L \log_2(\alpha_1) + \alpha_1^L \log_2(N-1) \nonumber \\
  & \leq LC_1 \alpha_1^L \nonumber \\
  & \leq C_2 \alpha^L ~,
  \label{eq:HSLxL_bound}
\vspace{-0.2in}
\end{align}
where $C_1 \equiv 2 - \log_2(\alpha_1) +  \log_2(N-1) > 0$,
step (a) follows from the fact that $1 - \alpha_1^L \geq 1/2$ for $L \geq L_1$,
and step (b) follows from the Taylor expansion bound on
$|\log_2(1 - \alpha_1^L)|$ for $L \geq L_2$.
In the last line, $\alpha$ may be chosen as any real number in the interval
$(\alpha_1,1)$ and $C_2 = C_2(\alpha)$ is chosen sufficiently large to ensure
the last inequality holds for all $L \geq L_0$.

Equation (\ref{eq:HSLxL_bound}) implies that, for all $L \geq L_0$:
\begin{align*}
\Prob(U_L \leq C_2 \alpha^L) & \geq \Prob(V_L^-) \geq 1 - K_1 \alpha_1^L ~.
\end{align*}
So, we know that, for all $L \geq L_0$:
\begin{align*}
\Prob(U_L > C_2 \alpha^L) \leq K_1 \alpha_1^L \leq K_1 \alpha^L ~.
\end{align*}
Therefore, for any $\alpha_2 \in (\alpha,1)$ and $L$ sufficiently large:
\begin{align*}
\Prob(U_L > \alpha_2^L) \leq K_1 \alpha_2^L ~.
\end{align*}
And, hence, there exists some $K_2 \geq K_1$ such that:
\begin{align*}
\Prob(U_L > \alpha_2^L) \leq K_2 \alpha_2^L ~, \mbox{ for all } L \in \N ~.
\end{align*}
\end{proof} 

\vspace{-0.1in}
To establish the theorem in general now, we show that for any machine $M$ there
exists a power machine $M^n$ satisfying Eq. (\ref{eq:StrongAssumption}).
To do so requires several additional lemmas.  

\begin{Lem} 
Let $M$ be a nonexact \eM. Then, for all $\ms$ and $\causalstate_j$
with $Pr(\ms|\causalstate_j) > 0 :$
\begin{align*}
g(\ms, \causalstate_j) \leq f(\ms, \causalstate_j)
	\frac{\Prob(\ms)}{\Prob(\causalstate_j)} \lambda^2 ~,
\end{align*}
where $\lambda \equiv \max_{i,j} \pi_i/\pi_j$ and $\Prob(\causalstate_j)$ and
$\Prob(\ms)$ are the respective stationary probabilities of the state
$\causalstate_j$ and symbol $\ms$: $\Prob(\causalstate_j )= \pi_j$ and
$\Prob(\ms)= \Prob(\MS_0 = \ms|\CS_0 \sim \pi)$. 
\label{lem:QPBound}
\end{Lem}

\begin{proof}
Fix $\ms$ and $\causalstate_j$.
Take $\causalstate_{k_1} \in \CausalStateSet_{\ms,j}$ such that:
\begin{align*}
\Prob(\causalstate_{k_1} | \ms) / \Prob(\causalstate_{k_1})
	= \max_{\causalstate_k \in \CausalStateSet_{\ms,j}}
	\Prob(\causalstate_{k} | \ms) / \Prob(\causalstate_k) ~,
\end{align*}
and take $\causalstate_{k_2} \in \CausalStateSet_{\ms,j}$ such that:
\begin{align*}
\Prob(\causalstate_{k_2} | \ms) = \max_{\causalstate_k \in \CausalStateSet_{\ms,j}}
	\Prob(\causalstate_{k} | \ms) ~.
\end{align*}
\vspace{-0.2in}
Then:
\begin{align*}
\frac{\Prob(\causalstate_{k_2}|\ms)}{\Prob(\causalstate_{k_2})} & \geq
	\frac{\Prob(\causalstate_{k_1}|\ms)}{\Prob(\causalstate_{k_2})} \\
	& = \frac{\Prob(\causalstate_{k_1}|\ms)}{\Prob(\causalstate_{k_1})} \cdot
	\frac{\Prob(\causalstate_{k_1})}{\Prob(\causalstate_{k_2})} \\
	& \geq \frac{\Prob(\causalstate_{k_1} | \ms)}{\Prob(\causalstate_{k_1})} \cdot
	1/\lambda ~,
\end{align*}
\vspace{-0.2in}
and
\begin{align*}
\frac{\Prob(\causalstate_{k_2}|\ms)}{\Prob(\causalstate_j)}
	& = \frac{\Prob(\causalstate_{k_2}|\ms)}{\Prob(\causalstate_{k_2})}
	\cdot \frac{\Prob(\causalstate_{k_2})}{\Prob(\causalstate_j)} \\
	& \geq \frac{\Prob(\causalstate_{k_2}|\ms)}{\Prob(\causalstate_{k_2})}
	\cdot 1/\lambda ~.
\end{align*}
Combining these relations we see that:
\begin{align*}
\frac{\Prob(\causalstate_{k_1} | \ms)}{\Prob(\causalstate_{k_1})}
	\leq \lambda^2
	\cdot  \frac{\Prob(\causalstate_{k_2} | \ms)}{\Prob(\causalstate_j)} ~.
\end{align*}
\vspace{-0.1in}
And, therefore:
\begin{align*}
g(\ms, \causalstate_j) & = \max_{\causalstate_k \in \CausalStateSet_{\ms,j}}
	\{ \Prob(\ms|\causalstate_k) \} \\
  & =  \max_{\causalstate_k \in \CausalStateSet_{\ms,j}}
  	\left\{ \Prob(\causalstate_k|\ms) \cdot \frac{\Prob(\ms)}{\Prob(\causalstate_k)} \right\} \\
  & = \max_{\causalstate_k \in \CausalStateSet_{\ms,j}}
  	\{ \Prob(\causalstate_k|\ms)/{\Prob(\causalstate_k)} \} \cdot \Prob(\ms) \\
  & = \frac{\Prob(\causalstate_{k1}|\ms)}{\Prob(\causalstate_{k1})}
  	\cdot \Prob(\ms) \\ 
  & \leq \lambda^2 \cdot
  	\frac{\Prob(\causalstate_{k2} | \ms)}{\Prob(\causalstate_j)} \cdot \Prob(\ms) \\
  & = f(\ms, \causalstate_j) \frac{\Prob(\ms)}{\Prob(\causalstate_j)}
         \lambda^2 ~.
\end{align*}
\end{proof}

\vspace{-0.2in}
Define $\Ae = \MeasAlphabet_{\epsilon,1} = \{ \ms \in \MeasAlphabet: u(\ms) < \epsilon \}$.
\begin{Lem} 
Let $M$ be a nonexact \eM\ such that:
\begin{enumerate}
\setlength{\topsep}{0mm}
\setlength{\itemsep}{0mm}
\item $\Prob(\Ae) > 1- \epsilon$, for some $\epsilon < 1/2$~, \mbox{ and }
\item $\Prob(\causalstate_{\max,\ms} | \ms)
	/\left(f(\ms, \causalstate_{\max,\ms}) \lambda^2 \right)
	\geq 2^{2 N^2 \lambda^2}$, for all $\ms \in \Ae$.
\end{enumerate}
Then,
$\Ex_{\pi_{edge}} \left\{ \log_2 \left( \frac{\Prob(\MS_0 | \CS_0)}{g(\MS_0, \CS_0)} \right) \right\} > 0$.
\label{lem:AEpsBound}
\end{Lem}

\begin{proof}
Applying Lemma \ref{lem:QPBound} we see:
\begin{align}
\Ex&_{\pi_{edge}} \left\{ \log_2 \left( \frac{\Prob(\MS_0 | \CS_0)}{g(\MS_0, \CS_0)} \right) \right\} 
  \nonumber \\
  & = \sumx \sum_j \Prob(\ms,\causalstate_j)
  \log_2 \left( \frac{\Prob(\ms|\causalstate_j)}{g(\ms, \causalstate_j)} \right) 
  \nonumber \\
  & \geq  \sumx \sum_j \Prob(\ms,\causalstate_j)
  \log_2 \left( \frac{\Prob(\causalstate_j|\ms)
  \Prob(\ms) / \Prob(\causalstate_j)}{f(\ms, \causalstate_j) \lambda^2
  \Prob(\ms)/\Prob(\causalstate_j)}   \right)  
  \nonumber \\
  & = \sumx \Prob(\ms) \sum_j \Prob(\causalstate_j|\ms)
  \log_2 \left(  \frac{\Prob(\causalstate_j|\ms)}
  {f(\ms, \causalstate_j) \lambda^2} \right)
  \nonumber \\
  & = \sum_{\ms \in \Ae} \Prob(\ms) \sum_j \Prob(\causalstate_j|\ms)
  \log_2 \left(
  \frac{\Prob(\causalstate_j|\ms)}{f(\ms, \causalstate_j) \lambda^2} \right)
  \nonumber \\
  &~~ + \sum_{\ms \in \Ae^c} \Prob(\ms) \sum_j \Prob(\causalstate_j|\ms)
  \log_2 \left( \frac{\Prob(\causalstate_j|\ms)}{f(\ms, \causalstate_j)
  \lambda^2} \right) ~.
\label{eq:i}
\end{align}

Now, for any $\ms \in \Ae^c$ we have:
\begin{align*}
\sum_j \Prob(\causalstate_j|\ms) & \log_2 \left(
  \frac{\Prob(\causalstate_j|\ms)}{f(\ms, \causalstate_j) \lambda^2}
  \right) \\
  & \geq \sum_j \lambda^2 \left[ \frac{\Prob(\causalstate_j|\ms)}{\lambda^2}
  \log_2 \left(  \frac{\Prob(\causalstate_j|\ms)}{\lambda^2} \right) \right] \\
  & \geq \sum_j \lambda^2 \cdot
  	-H \left( \frac{\Prob(\causalstate_j | \ms)}{\lambda^2} \right) \\
  & \geq \sum_j \lambda^2 \cdot (-1) \\
  & = - N \lambda^2 ~,
\end{align*}
where $H(\cdot)$ is the binary entropy function. So:
\begin{align}
\sum_{\ms \in \Ae^c} \Prob(\ms) \sum_j \Prob(\causalstate_j|\ms) & \log_2
  \left(  \frac{\Prob(\causalstate_j|\ms)}{f(\ms, \causalstate_j)
  \lambda^2} \right)
  \nonumber \\
  & \geq \Prob(\Ae^c) \cdot -N \lambda^2
  \nonumber \\
  & > - \epsilon N \lambda^2 ~.
\label{eq:II}
\end{align}

Also, if we let $\CausalStateSet^- \equiv
\CausalStateSet/\{\causalstate_{\max,\ms}\}$, then for any $\ms \in \Ae$:
\begin{align*}
  \sum_j & \Prob(\causalstate_j|\ms) \log_2 \left(\frac{\Prob(\causalstate_j|\ms)}{f(\ms, \causalstate_j) \lambda^2} \right) \\
  & = ~ \Prob(\causalstate_{\max,\ms} | \ms) \log_2
  \left(\frac{\Prob(\causalstate_{\max,\ms}|\ms)}{f(\ms,
  \causalstate_{\max,\ms}) \lambda^2} \right) \\
  & ~~~~~ + \sum_{\causalstate_j \in \CausalStateSet^-} \Prob(\causalstate_j|\ms) \log_2 \left(  \frac{\Prob(\causalstate_j|\ms)}{f(\ms, \causalstate_j) \lambda^2} \right) \\
  & \geq \frac{1}{N} 2 N^2 \lambda^2
  	+  \sum_{\causalstate_j \in \CausalStateSet^-} \Prob(\causalstate_j|\ms) 
      \log_2\left(  \frac{\Prob(\causalstate_j|\ms)}{f(\ms, \causalstate_j) \lambda^2} \right) \\
  & \geq \frac{1}{N} 2 N^2 \lambda^2 +  \sum_{\causalstate_j \in \CausalStateSet^-}  \lambda^2 \left[ \frac{\Prob(\causalstate_j|\ms)}{\lambda^2}
     \log_2  \left(  \frac{\Prob(\causalstate_j|\ms)}{\lambda^2} \right) \right] \\
  & \geq 2 N \lambda^2 + \sum_{\causalstate_j \in \CausalStateSet^-} \lambda^2 \cdot - H\left(\frac{\Prob(\causalstate_j|\ms)}{\lambda^2} \right) \\
  & \geq 2 N \lambda^2 - N \lambda^2 \\
  & = N \lambda^2 ~.
\end{align*}
And, hence:
\begin{align}
\sum_{\ms \in \Ae} \Prob(\ms) \sum_j \Prob(\causalstate_j|\ms) \log_2
  & \left(  \frac{\Prob(\causalstate_j|\ms)}{f(\ms, \causalstate_j)
  \lambda^2} \right)
  \nonumber \\
  & \geq \Prob(\Ae) \cdot N \lambda^2
  \nonumber \\
  & > (1-\epsilon) N \lambda^2 ~.
\label{eq:III}
\end{align}

Combining Eqs. (\ref{eq:i}), (\ref{eq:II}), and (\ref{eq:III}), we see that:
\begin{align*}
\Ex \left\{ \log_2 \left( \frac{\Prob(\MS_0 | \CS_0)}{g(\MS_0, \CS_0)} \right)  \right\}
  > (1-2 \epsilon) N \lambda^2 \equiv C' ~,
\end{align*}
where $C' > 0$ for $\epsilon < 1/2$.
Since $M$ is not exactly synchronizable, we know $g(\ms, \causalstate_j)$ is 
always strictly positive, so this expectation must be finite. Hence, there
exists some real number $C > C' > 0$ such that:
\begin{align*}
\Ex \left\{ \log_2 \left( \frac{\Prob(\MS_0 | \CS_0)}{g(\MS_0, \CS_0)} \right)  \right\} = C ~.
\end{align*}
\end{proof}

\begin{Rem}
In the above proof we implicitly assumed $\Prob(\ms,\causalstate_j) \not= 0$ 
for all $\ms$ and $j$. The sums for the expectation are, of course, computed
only over those $\ms$ and $j$ for which $\Prob(\ms,\causalstate_j) \not= 0$.
Terms involving pairs $(x,\causalstate_j)$ with $\Prob(x,\causalstate_j) = 0$
should be omitted.
\end{Rem}

\begin{Lem} For any nonexact \eM\ $M$, there exists some $n \in \N$ such that 
the power machine $M^n$ is an \eM\ with
\begin{align*}
\Ex \left\{ \log_2 \left( \frac{\Prob(Y_0 | R_0)}{g(Y_0, R_0)} \right) \right\} > 0 ~,
\end{align*}  
where $Y_L$ is the RV for the the $L_{th}$ output symbol generated by the
machine $M^n$ and $R_L$ is the RV for the $L_{th}$ $M^n$-state.
\label{lem:ExpectationBound}
\end{Lem}
We also denote the alphabet of $M^n$ as $\B$ and the set $\Ae$ for the machine
$M^n$ as $B_{\epsilon}$; i.e., $B_{\epsilon} \sim \Aen$ for $M$. We define
$\causalstateb(\phi)$ to be the most likely state in a distribution $\phi$
over the machine states, and $\Prob(\causalstateb(\phi))$ to be the probability
of this state in the distribution $\phi$. For example, if
$\phi = (0.3,0.1,0.2,0.4)$, then 
$\causalstateb(\phi) = \causalstate_4$ and $\Prob(\causalstateb(\phi)) = 0.4$. 

\begin{proof}
Given any nonexact \eM\ $M$, 
\begin{enumerate} 
\setlength{\topsep}{0mm}
\setlength{\itemsep}{0mm}
\item Take $\epsilonb = 1 / \left( N \lambda^2 2^{2 N^2 \lambda^2} \right) $. 
\item Take $\epsilon$ small enough that
	$\Prob(\causalstateb(\phi)) > 1 - \epsilonb$ 
	for any state distribution $\phi$ with $H[\phi] < \epsilon$.
	(Without loss of generality, we may assume $\epsilon < 1/2$.)
\item For $\epsilon$ as above, take $n$ relatively prime to the period 
	$p$ of $M$'s graph and large enough such that
	$\Prob(\Aen) > 1 - \epsilon$.
	(Note that this is possible since
	$\lim_{L \to \infty} \Prob(\AeL) = 1$, for all $\epsilon > 0$,
	since $\lim_{L \to \infty} \AvgUncertainty(L) = 0$.)
\end{enumerate}
Then, $M^n$ is an \eM\ for the process $\Process^n$ and
$\Prob(B_{\epsilon})  = \Prob(\Aen) > 1 - \epsilon$.
Moreover, for all $y \in B_{\epsilon}$ we have:
\begin{align*}
H[\phi(y)] < \epsilon 	& \stackrel{(a)}{\implies} \Prob(\causalstateb(\phi(y))) > 1 - \epsilonb \\
				& \implies f(y,\causalstate_{\max,y}) < \epsilonb \\
				& \implies \frac{\Prob(\causalstate_{\max,y} |
				y)}{f(y,\causalstate_{\max,y})}  > \frac{1/N}{\epsilonb} \\
				& \stackrel{(b)}{\implies}
				\frac{\Prob(\causalstate_{\max,y} |
				y)}{f(y,\causalstate_{\max,y})}  >\lambda^2 2^{2 N^2 \lambda^2} \\
				& \implies \frac{\Prob(\causalstate_{\max,y} |
				y)}{f(y,\causalstate_{\max,y}) \lambda^2} > 2^{2 N^2 \lambda^2} ~, \\
\end{align*}
where step (a) follows from item 2 above and step (b) follows from our choice
of $\epsilonb$. Hence, by Lemma \ref{lem:AEpsBound}:
\begin{align*}
\Ex \left\{ \log_2 \left( \frac{\Prob(Y_0 | R_0)}{g(Y_0, R_0)} \right)
  \right\} = C > 0 ~.  
\end{align*}
(Note that $\lambda$ for $M$ is the same as  $\lambda$ for $M^n$,
since $M$ and $M^n$ have the same stationary distribution $\pi$.)
\end{proof}

Finally, in order to convert between $U_L$ and $Q_L$ convergence in
our theorem we need one last lemma. 

\begin{Lem} For any $\Phi_L$:
\begin{enumerate}
\setlength{\topsep}{0mm}
\setlength{\itemsep}{0mm}
\item If $Q_L \leq 1/2$, then $U_L \geq Q_L$. 
\item If $Q_L > 1/2$, then $U_L \geq H\left( 1 / N \right)$, where
	$H(\cdot)$ is the binary entropy function.
\end{enumerate}
\label{Lem:QLUL}
\end{Lem}

\begin{proof}  Note that: 
\begin{align*}
U_L & = H[\Phi_L] \\
  & \geq H[(1- Q_L, Q_L, 0, \ldots ,0)] \\
  & = H(Q_L) ~.
\end{align*}
Since $H(Q_L) \geq Q_L \log_2 \left(1 / Q_L \right)$,
we know $H(Q_L) \geq Q_L$ for $Q_L \leq 1/2$. Since $H(Q_L)$
is monotonically decreasing on $[\frac{1}{2}, 1 - 1 / N]$ 
and $Q_L$ is at most $1 - 1 / N$, we know
$H(Q_L) \geq H(1- 1/N) = H(1/N)$ for $Q_L > 1/2$.  
\end{proof}

Using these lemmas we can now prove the primary result of this section.

\begin{proof} (Nonexact Machine Synchronization Theorem)
We first prove Claim 2 of the theorem and then use this to show Claim 1.

\emph{Proof of Claim 2}: Given any nonexact \eM\ $M$, take a power
machine $M^n$ as in Lemma \ref{lem:ExpectationBound} such that:
\begin{align*}
\Ex \left\{ \log_2 \left( \frac{\Prob(Y_0|R_0)}{g(Y_0,R_0)} \right) \right\} = C > 0  ~.
\end{align*}
Denote the random variable $U_L$ for the machine $M^n$ as $V_L$ and the quantity
$\AvgUncertainty(L)$ for the machine $M^n$ as $\V(L)$. By Lemma 
\ref{lem:ConstantsAssumingExpectation} we know there exist constants $B_1 > 0$
and $0 < \eta_1 < 1$ such that:
\begin{align*}
\Prob(V_L > \eta_1^L) \leq B_1 \eta_1^L ~, \mbox{ for all $L \in \N$.}
\end{align*}
A proof identical to that of Prop. \ref{prop:StateUncertaintyConverge} below
then shows there exists some $B_2 > B_1$ such that:
\begin{align*}
\V(L) \leq B_2 \eta_1^L, \mbox{ for all $L \in N$.}
\end{align*}
Or, equivalently:
\begin{align*}
\AvgUncertainty(nL) \leq B_2 \eta_1^L ~, \mbox{ for all $L \in \N$.}
\end{align*}
Taking  $\eta_2 = \eta_1^{1/n}$ we have:
\begin{align*}
\AvgUncertainty(m) \leq B_2 \eta_2^m ~,
\end{align*}
for any length $m$ that is an integer multiple of $n$. 
Since $\AvgUncertainty(m) \leq \log_2(N)$ for any $m$ and is
monotonically decreasing, it follows that:
\begin{align*}
\AvgUncertainty(m) \leq K \eta_2^m, \mbox{ for all $m \in \N$,}
\end{align*}
where $K \equiv \max \{ B_2, \log_2(N) \} / \eta_2^n$. 
Thus, by Markov's inequality, we know that for any $m \in \N$ and $t > 0$: 
\begin{align*}
\Prob(U_m > t) 	& \leq \frac{\Ex U_m}{t} = \frac{\AvgUncertainty(m)}{t} \leq \frac{K \eta_2^m}{t} ~.
\end{align*}
Taking $t = \eta_2^{m/2}$ yields:
\begin{align*}
\Prob(U_m > \alpha^m) \leq K \alpha^m ~,
\end{align*}
where $\alpha \equiv \eta_2^{1/2}$. \\

\emph{Proof of Claim 1}: By Claim 2 we know there exist constants $K > 0$ and
$0 < \alpha < 1$ such that:
\begin{align*}
\Prob(U_L > \alpha^L) \leq K \alpha^L, \mbox{ for all } L \in \N. 
\end{align*} 
Take $L_0$ large enough that $\alpha^{L_0} < H(1/N)$. 
Then, for all $L \geq L_0$ we have:
\begin{align*}
\Prob(Q_L & > \alpha^L)	= \Prob(\alpha^L < Q_L \leq 1/2) + \Prob(Q_L > 1/2) \\
	& \stackrel{(*)}{\leq} \Prob(U_L > \alpha^L, Q_L \leq 1/2) + \Prob(U_L \geq H(1/N)) \\
	& \leq \Prob(U_L > \alpha^L) + \Prob(U_L > \alpha^L) \\
	& \leq 2K \alpha^L ~,
\end{align*}
where step (*) follows from Lemma \ref{Lem:QLUL}. Hence, for some
$\Kt \geq 2K$ we have:
\begin{align*}
\Prob(Q_L > \alpha^L) \leq \Kt \alpha^L ~, \mbox{ for all } L \in \N ~.
\end{align*}
\end{proof}

\vspace{-0.3in}
\section{Consequences}
\vspace{-0.1in}
\label{sec:EntropyConv}

As a direct consequence of Thm. \ref{the:NESyncRate} we establish exponential
convergence results for $\AvgUncertainty(L)$ and $\hmu(L)$ analogous to those
in the exact case \cite{Trav10a}. We also use Thm. \ref{the:NESyncRate} 
to prove the existence of pointwise almost everywhere (a.e.) exponential
synchronization for nonexact machines. This establishes that any \eM\ is
indeed asymptotically synchronizable in the pointwise sense of
Def. \ref{Def:AsymptoticSync}. 

\vspace{-0.2in}
\subsection{Exponential Convergence of  $\AvgUncertainty(L)$}
\vspace{-0.1in}

\begin{Prop} 
For any nonexact \eM\ $M$ there exist constants $K > 0$ and $0 < \alpha < 1$
such that
\begin{align*}
\AvgUncertainty(L) \leq K \alpha^L ~, \mbox{ for all } L \in \N ~.
\end{align*}
\label{prop:StateUncertaintyConverge}
\end{Prop}
\vspace{-0.3 in}
\begin{proof}
Let $M$ be any nonexact \eM. Then by Thm. \ref{the:NESyncRate} there exist
constants $C > 0$ and $0 < \alpha < 1$ such that
$\Prob(U_L > \alpha^L) \leq C \alpha^L$, for all $L \in \N$. Define:
\begin{align*}
A_L  		& = \{ w \in \LLM: u(w) \leq \alpha^L \} \mbox{ and }\\
A_L^c 	& = \LLM/A_L
\end{align*}
\vspace{-0.1in}
Then,
\begin{align*}
\AvgUncertainty(L) & = \sumw \Prob(w) u(w) \\
  & = \sum_{w \in A_L} \Prob(w) u(w)
  	+ \sum_{w \in A_L^c} \Prob(w) u(w) \\
  & \leq \Prob(A_L) \cdot \alpha^L + \Prob(A_L^c) \cdot \log_2(N) \\
  & \leq 1 \cdot \alpha^L + C \alpha^L \cdot \log_2(N) \\
  & = K \alpha^L ~,
\end{align*}
where $K \equiv 1 + C \log_2(N)$.
\end{proof}

\vspace{-0.2in}
\subsection{Exponential Convergence of  $\hmu(L)$}
\vspace{-0.1in}

\begin{Prop}
For any nonexact \eM\ $M$, there exist constants $K > 0$ and $0 < \alpha < 1$
such that:
\begin{align*}
\hmu(L) - \hmu \leq K \alpha^L ~, \mbox{ for all } L \in \N ~.
\end{align*}
\end{Prop}

\vspace{-0.2in}
\begin{proof}
This follows directly from Prop. \ref{prop:StateUncertaintyConverge} and
Lemma \ref{lem:AvgUnceraintyL_boundon_hmuL} below.
\end{proof}

\vspace{-0.2in}
\begin{Lem}
For any \eM\ $M$ and any $L \in \N$:
\begin{align}
\hmu(L+1) - \hmu \leq \AvgUncertainty(L) ~.
\end{align}
\label{lem:AvgUnceraintyL_boundon_hmuL}
\end{Lem}

\vspace{-0.2in}
\begin{proof}
Note that:
\begin{align}
\label{eq:Expansion1}H[\Future^L,\MS_L,\CS_L]
  & = H[\Future^L] + H[\CS_L|\Future^L] + H[\MS_L|\Future^L,\CS_L] \nonumber \\
  & =H[\Future^L] + H[\CS_L|\Future^L]  + H[\MS_L|\CS_L] \nonumber \\
  & = H[\Future^L] + H[\CS_L|\Future^L] + \hmu ~,
\end{align}
and also that:
\begin{align}
\label{eq:Expansion2}
H[\Future^L,\MS_L,\CS_L] & = H[\Future^L] + H[\MS_L|\Future^L] \nonumber \\
  & ~~~~~~~ + H[\CS_L|\Future ^L,\MS_L] ~.
\end{align}
Equating the RHS of Eqs. (\ref{eq:Expansion1}) and (\ref{eq:Expansion2}) gives:
\begin{align}
H[\CS_L|\Future^L] + \hmu
  & = H[\MS_L|\Future^L] + H[\CS_L|\Future^L,\MS_L] \nonumber \\
  & \geq H[\MS_L|\Future^L] ~.
\end{align}
\vspace{-0.1in}
Or, in other words:
\begin{align}
\AvgUncertainty(L) + \hmu \geq \hmu(L+1) ~.
\end{align}
\end{proof}

\vspace{-0.2in}
\begin{Rem}
If we define the \emph{synchronization} and \emph{predication decay constants},
respectively, as:
\begin{align*}
\alpha_s & = \limsup_{L \to \infty} ~\AvgUncertainty(L)^{1/L} \\
\alpha_p & = \limsup_{L \to \infty} \left(\hmu(L) - \hmu \right)^{1/L} ~,
\end{align*}
then Lemma \ref{lem:AvgUnceraintyL_boundon_hmuL} also implies
that $\alpha_p \leq \alpha_s$. This is to say, the observer's predictions
approach their optimal level at least as fast as the observer synchronizes.
Since Lemma \ref{lem:AvgUnceraintyL_boundon_hmuL} applies to any \eM, this
statement also holds for any \eM\ (exact or nonexact).
\end{Rem}

\subsection{Pointwise a.e. Asymptotic Synchronization}  
\label{sec:PointwiseSync}

\begin{Prop}
For any nonexact \eM\ $M$ there exists some $0 < \alpha < 1$ such that for
a.e. $\future \in \LiM$, there exists $L_0 \in \N$ such that for all
$L \geq L_0$,
\begin{align*}
u(\future^L) \leq \alpha^L ~.
\end{align*}
\end{Prop}

\begin{proof}
Apply the Borel-Cantelli Lemma to Thm. \ref{the:NESyncRate}. 
\end{proof}

\section{Conclusion} 
\label{sec:Concl}

We analyzed the process of asymptotic synchronization to nonexact \eMs.
Although the treatment is more involved mathematically, the primary results are
essentially the same as those for the exact case given in Ref. \cite{Trav10a}.
An observer's average state uncertainty $\AvgUncertainty(L)$ vanishes
exponentially fast and, consequently, an observer's average uncertainty
in predictions $\hmu(L)$ converges to the machine's entropy rate $\hmu$
exponentially fast, as well. 

We hope to extend the asymptotic synchronization results to more general model
classes such as countable-state \eMs\ or nonunifilar HMMs. We also intend to
improve the bounds on the constant $\alpha$ given in the convergence theorems.

\section*{Acknowledgments}

NT was partially supported by an NSF VIGRE fellowship. This work was partially
supported by the Defense Advanced Research Projects Agency (DARPA) Physical
Intelligence project via subcontract No. 9060-000709. The views, opinions, and
findings contained here are those of the authors and should not be interpreted
as representing the official views or policies, either expressed or implied,
of the DARPA or the Department of Defense.

\bibliography{ref,chaos,other}

\end{document}

%% file: cmechabbrev.tex
\newcommand{\eM}     {$\epsilon$\protect\nobreakdash-machine}
\newcommand{\eMs}    {$\epsilon$\protect\nobreakdash-machines}

\newcommand{\EMs}    {$\epsilon$\protect\nobreakdash-Machines}

\newcommand{\Process}{\mathcal{P}}

\newcommand{\MeasAlphabet}	{\mathcal{A}}
\newcommand{\MeasSymbol}   { {X} }
\newcommand{\meassymbol}   { {x} }

\newcommand{\past}	{ {\overleftarrow {\meassymbol}} }

\newcommand{\Future}	{ \overrightarrow{\MeasSymbol} }
\newcommand{\future}	{ \overrightarrow{\meassymbol} }

\newcommand{\CausalState}	{ \mathcal{S} }

\newcommand{\causalstate}	{ \sigma }

\newcommand{\CausalStateSet}	{ \boldsymbol{\CausalState} }

\newcommand{\Prob}      {\Pr} 

\newcommand{\hmu}		{h_\mu}



%% file: shortcuts.tex
\theoremstyle{plain}   
\theoremstyle{plain}   \newtheorem{Lem}{Lemma}
\theoremstyle{plain} 	
\theoremstyle{plain} 	\newtheorem{The}{Theorem}
\theoremstyle{plain} 	\newtheorem{Prop}{Proposition}
\theoremstyle{plain} 	
\theoremstyle{plain}	\newtheorem*{Rem}{Remark}
\theoremstyle{plain}	\newtheorem{Def}{Definition} 
\theoremstyle{plain}	

\def\norm#1{\|#1\|}

\def\implies{\Longrightarrow}

\def\N{\mathbb{N}} 

\def\R{\mathbb{R}}

\def\Ex{\mathbb{E}}

\def\vY{\overrightarrow{Y}}

\def\vz{\overrightarrow{z}}

\def\epsilonb{\overline{\epsilon}}

\def\th{\tilde{h}}

\def\Kt{\tilde{K}}

\def\sb{\overline{s}}

\def\Yb{\overline{Y}}

\def\futurestates{\overrightarrow{s}}

\def\causalstateb{\overline{\sigma}}
\def\AvgUncertainty{\mathcal{U}}
\def\V{\mathcal{V}}

\def\goesonx{\stackrel{x}{\rightarrow}}
\def\goesonw{\stackrel{w}{\rightarrow}} 
\def\B{\mathcal{B}}
\def\CS{\mathcal{S}}
\def\CSb{\overline{\mathcal{S}}}
\def\csr{s}
\def\cs{\sigma}
\def\MS{X}
\def\ms{x}

\newcommand{\SYN}{ \mathrm{SYN} }
\newcommand{\WSYN}{ \mathrm{WSYN} }

\def\L{\mathcal{L}}
\def\LM{\mathcal{L}(M)}

\def\LLM{\mathcal{L}_L(M)}
\def\LiM{\mathcal{L}_{\infty}(M)}

\def\AeL{A_{\epsilon,L}}
\def\Aen{A_{\epsilon,n}}
\def\AeLi{A_{\epsilon,L_i}}
\def\Ae{A_{\epsilon}}

\def\sumx{\sum_{x \in \MeasAlphabet}}
\def\sumw{\sum_{w \in \LLM}}

